\documentclass[12pt,a4paper]{article}

\title{Convex Hedging in Incomplete Markets}

\author{Birgit Rudloff \thanks{Princeton University, Department of
    Operations Research and Financial Engineering, Princeton, NJ 08544, USA,
    email: brudloff@princeton.edu
    \newline
    This work was partially supported by Vienna Science and Technology Fund (WWTF) grant P15889-G05.}}

\usepackage{amsmath, amssymb}
\usepackage{amsthm}
\usepackage{natbib}

\newtheorem{theorem}{Theorem}[section]

\newtheorem{remark1}[theorem]{Remark}
\newtheorem{lemma}[theorem]{Lemma}
\newtheorem{definition}[theorem]{Definition}

\newtheorem{assumption}[theorem]{Assumption}

\newenvironment{remark}{\begin{remark1}\rm}{\end{remark1}}

\newcommand{\of}[1]{\ensuremath{\left( #1 \right)}}

\newcommand{\R}{\mathrm{I\negthinspace R}}
\DeclareMathOperator*{\core}{core}
\DeclareMathOperator*{\dom}{dom}

\begin{document}
\date{April 3, 2007}
\maketitle

\begin{abstract}
In incomplete financial markets not every contingent claim can be
replicated by a self-financing strategy. The risk of the resulting
shortfall can be measured by convex risk measures, recently
introduced by \citet{FoeSch02}. The dynamic optimization problem of
finding a self-financing strategy that minimizes the convex risk of
the shortfall can be split into a static optimization problem and a
representation problem. It follows that the optimal strategy
consists in superhedging the modified claim $\widetilde{\varphi}H$,
where $H$ is the payoff of the claim and $\widetilde{\varphi}$ is
the solution of the static optimization problem, the optimal
randomized test.
\\
In this paper, we will deduce necessary and
sufficient optimality conditions for the static problem using
convex duality methods. The solution of the static optimization
problem turns out to be a randomized test with a typical
$0$-$1$-structure.
\\[.2cm]
\\[.2cm]
\\[.2cm]
\end{abstract}

\section{Introduction}
In an incomplete financial market not every contingent claim is
attainable and the equivalent martingale measure is no longer
unique. Thus, a perfect hedge as in the Black-Scholes-Merton model
is not possible any longer. Therefore, we are faced with the
problem of searching strategies which reduce the risk of the
resulting shortfall as much as possible.
\\
One can still stay on the safe side using a 'superhedging'
strategy. But from a practical point of view, the cost of
superhedging is often too high.
\\
For this reason, we consider the possibility of investing less
capital than the superhedging price of the liability. This leads to
a shortfall, the risk of which, measured by a suitable risk measure,
should be minimized. This problem has been studied using the
probability \citep{FoeLeu99}, the expectation of a loss function
\citep{FoeLeu00} and a coherent risk measure
\citep{Nakano03,Nakano04,Rudloff05CohH} to quantify the shortfall
risk. In our approach, we use convex risk measures, a generalization
of coherent risk measures. In analogy to the problems mentioned
above, the dynamic optimization problem of finding an admissible
strategy that minimizes the convex shortfall risk can be split into
a static optimization problem and a representation problem. The
optimal strategy consists in superhedging a modified claim
$\widetilde{\varphi}H$, where $H$ is the payoff of the claim and
$\widetilde{\varphi}$ is a solution of the static optimization
problem, an optimal randomized test. In this paper, we show the
existence of a solution to the static problem. We deduce with
methods of Fenchel duality a necessary and sufficient condition for
an optimal randomized test that gives a result about the structure
of the solution. The results are new and even improve, when
restricted to coherent risk measures, the results obtained by
\citet{Nakano03,Nakano04} \citep[see][Section~4.1.3]{Rudloff06Diss}.
\\
This paper is organized as follows: In Section~\ref{sec2}, we state
the formulation of the shortfall problem and review the definition
of convex risk measures. Then, we prove the possibility to decompose
the dynamic optimization problem into a static and a representation
problem. In Section~\ref{sec4}, we analyze the static optimization
problem. We show the existence of a solution, formulate the dual
problem and show that strong duality holds. Thus, the optimal
solution is a saddle point of a functional, specified in
Section~\ref{secSP}. To solve the problem, we first consider the
inner problem in Section~\ref{sec4.3} and deduce an extended
Neyman-Pearson lemma. Then, we solve the saddle point problem in
Section~\ref{sec4.4}. The optimal solution of the static
optimization problem is a randomized test with the typical
$0$-$1$-structure.

\section{Formulation of the Problem}\label{sec2}
The discounted price process of the $d$ underlying assets is
described as an $\R^d$-valued semimartingale
$S=(S_{t})_{t\in[0,T]}$ on a complete probability space
$(\Omega,\mathcal{F},P)$ with filtration
$(\mathcal{F}_{t})_{t\in[0,T]}$. The filtration is supposed to
satisfy the usual conditions (see e.g. \cite{KarShr98}). We write
$L^{1}$ and $L^{\infty}$ for $L^{1}(\Omega,\mathcal{F},P)$ and
$L^{\infty}(\Omega,\mathcal{F},P)$, respectively. We endow $L^{1}$
and $L^{\infty}$  with the norm topology. Then, the topological
dual space of $L^1$  can be identified with $L^{\infty}$ and the
topological dual space of $L^{\infty}$ can be identified with
$ba(\Omega,\mathcal{F},P)$, the space of finitely additive set
functions on $(\Omega,\mathcal{F})$ with bounded variation,
absolutely continuous to $P$ \citep[see][Chapter~IV, 9,
Example~5]{Yosida80}. Let $\langle X,Y\rangle$ be the bilinear
form between $L^{\infty}$ and $ba(\Omega,\mathcal{F},P)$: $\langle
X,Y\rangle=\int_\Omega YdX$ for all $Y\in L^{\infty}$, $X\in
ba(\Omega,\mathcal{F},P)$. For $X\in L^{1}$ it reduces to $\langle
X,Y\rangle=E[XY]$, where $E$ denotes the mathematical expectation
with respect to $P$.
\\
Let $\widehat{\mathcal{Q}}$ be the set of all probability measures
on $(\Omega,\mathcal{F})$ absolutely continuous with respect to
$P$. For $Q\in\widehat{\mathcal{Q}}$ we denote the expectation
with respect to $Q$ by $E^{Q}$ and the Radon-Nikodym derivative
$dQ/dP$ by $Z_{Q}$.
\\
An $\R^d$-valued semimartingale $S=(S_{t})_{t\in[0,T]}$ is called a
sigma-martingale if there exists an $\R^d$-valued martingale $M$ and
an $M$-integrable predictable $\R_+$-valued process $\xi$ such that
$S_t=\int_0^t \xi_s dM_s$, $t\in[0,T]$ (see \cite{DelSchach06},
Section~14.2). Let $\mathcal{P}_\sigma$ denote the set of
probability measures $P^*$ equivalent to $P$ such that $S$ is a
sigma-martingale with respect to $P^*$. Since we want to exclude
arbitrage opportunities, we assume that
$\mathcal{P}_\sigma\neq\emptyset$. To be more concrete: $S$
satisfies the condition of 'no free lunch with vanishing risk' if
and only if $\mathcal{P}_\sigma\neq\emptyset$ (Theorem~14.1.1,
\cite{DelSchach06}). The concept of 'no free lunch with vanishing
risk' is a mild strengthening of the concept of 'no arbitrage' that
has to be used in general semimartingale models. In the case of a
finite probability space $\Omega$ the above assertion holds true if
one replaces the term 'no free lunch with vanishing risk' by the
term 'no arbitrage' and the set $\mathcal{P}_\sigma$ by the set of
equivalent martingale measures $\mathcal{P}$. In the case of an
$\R^d$-valued (locally) bounded semimartingale $S$ one may replace
the set $\mathcal{P}_\sigma$ by the set of equivalent (local)
martingale measures $\mathcal{P}_{(loc)}$ (see \cite{DelSchach06},
Chapter~9).
\\
Equations and inequalities between random variables are always
understood as $P-a.s.$
\\
A self-financing strategy is given by an initial capital $V_{0}\geq 0$ and
a predictable process $\xi$ such that the resulting value
process
\[
    V_{t}=V_{0}+\int_{0}^{t}\xi_{s}dS_{s},\quad  t\in[0,T],
\]
is well defined. Such a strategy $(V_{0},\xi)$ is called admissible
if the corresponding value process $V_t$ satisfies $V_{t}\geq 0$ for
all $t\in[0,T]$.
\\
Consider a contingent claim. Its payoff is given by an
$\mathcal{F}_{T}$ -\mbox{measurable}, nonnegative random variable
$H\in L^{1}$. We assume
\begin{equation}
    \label{SHP}
    U_{0}=\sup_{P^{*}\in\mathcal{P}_\sigma}E^{P^{*}}[H]<+\infty.
\end{equation}
The above equation is the dual characterization of the superhedging
price $U_{0}$, the smallest amount $V_{0}$ such that there exists an
admissible strategy $(V_{0},\xi)$ with value process $V_{t}$
satisfying $V_{T}\geq H$ (see \cite{DelSchach06}, Theorem~14.5.20).
The corresponding strategy is called the superhedging strategy of
the claim $H$. Again, in the case where $S$ is an $\R^d$-valued
(locally) bounded semimartingale, one may replace the set
$\mathcal{P}_\sigma$ by the set of equivalent (local) martingale
measures $\mathcal{P}_{(loc)}$. In the complete case, where the
equivalent sigma-martingale measure $P^{*}$ is unique,
$U_{0}=E^{P^{*}}[H]$ is the unique arbitrage-free price of the
contingent claim.
\\
Since superhedging can be quite expensive in the incomplete market
(see e.g. \cite{GusMor02}), we search for the best hedge an
investor can achieve with a smaller amount
$\widetilde{V}_{0}<U_{0}$. In other words, we look for an
admissible strategy $(V_{0},\xi)$ with
$0<V_{0}\leq\widetilde{V}_{0}$ that minimizes the risk of losses
due to the shortfall $\{\omega:V_{T}(\omega)<H(\omega)\}$, this
means we want to minimize the risk of $-(H-V_T)^+$. The risk will
be measured by a convex risk measure $\rho$, recently introduced
by \citet{FoeSch02}. Thus, we consider the dynamic optimization
problem of finding an admissible strategy that solves
\begin{equation}
    \label{dynop1}
    \min_{\scriptstyle{(V_{0},\xi)}}\rho\of{-\of{H-V_{T}}^{+}}
\end{equation}
under the capital constraint of investing less capital than the
superhedging price
\begin{equation}
    \label{dynop2}
    0<V_{0}\leq\widetilde{V}_{0}<U_{0}.
\end{equation}
For the convenience of the reader we recall the definition and some
properties of convex risk measures. In contrast to \citet{FoeSch02},
where $\rho:\;L^{\infty}\to\R$, we consider convex risk measures
defined on $L^1$ that can also attain the value $+\infty$ for
investments that are not acceptable in any way.

\begin{definition}[convex risk measure]\label{def krm}
A function $\rho:\;L^{1}\to\R\cup\{+\infty\}$ with $\rho(0)=0$ is
a convex risk measure if it satisfies for all $X_1,X_2\in L^{1}$:
\begin{itemize}

\item[(i)]monotonicity: $\quad X_1 \geq X_2\Rightarrow \rho
\of{X_1} \leq \rho \of{X_2}$,

\item[(ii)]translation property: $\quad c\in\R\Rightarrow\rho
\of{X_1 + c\textbf{1}} = \rho \of{X_1} - c$,

\item[(iii)]convexity: $\quad \lambda \in [0,1]\Rightarrow\rho
\of{\lambda X_1+(1-\lambda)X_2} \leq \lambda\rho \of{X_1} +
(1-\lambda)\rho \of{X_2}$.
\end{itemize}
\end{definition}
\noindent The random variable equal to $1$ almost surely is denoted
by $\textbf{1}$ in (ii). The assumption $\rho(0)=0$ is reasonable
and ensures that $\rho(X)$ can be interpreted as risk adjusted
capital requirement. The set $\mathcal A:=\{X\in L^1:\;\rho(X)\leq
0\}$ is called the acceptance set of the risk measure $\rho$. It is
well known that each lower semicontinuous convex risk measures
admits a dual representation (see \cite{FoeSch02} for risk measures
on $L^\infty$ and for those on general $L^p$ spaces, see for
instance \cite{RusSha06}, \cite{Ham05}, \cite{Rudloff06Diss}). A
function $\rho:\;L^{1}\to\R\cup\{+\infty\}$ is a lower
    semicontinuous, convex risk measure if and only if there exists a
    representation of the form
    \begin{equation}
        \label{dkrm1}\rho\of{X}=\sup_{Q\in \mathcal{Q}}\{E^Q[-X]-\sup_{\widetilde{X}\in\mathcal A}E^Q[-\widetilde{X}]\},
    \end{equation}
    where $\mathcal Q:=\{Q\in\widehat{\mathcal Q}:Z_Q\in L^\infty\}$.
    The conjugate function $\rho^*$ of $\rho$ is
    nonnegative, convex, proper, weakly* lower semicontinuous,
    \begin{equation*}
        \dom \rho^*\subseteq \{-Z_Q:Q\in\mathcal Q\}
    \end{equation*}
    and $\rho^*$ satisfies for all $Y\in L^\infty$ with $E[Y]=-1$
    \begin{equation}\label{r*}
        \rho^*(Y)=\sup_{X\in\mathcal A}E[XY].
    \end{equation}

\begin{remark}\label{remark alpha}
Let $\alpha(Q):\mathcal Q\rightarrow\R\cup\{+\infty\}$ be a
functional with $\inf_{Q\in\mathcal Q}\alpha(Q)=0$. Then,
 \begin{equation*}
    \rho\of{X}:=\sup_{Q\in\mathcal{Q}}\{E^{Q}[-X]-\alpha(Q)\}
 \end{equation*}
is a convex risk measure. The functional $\alpha$ is called a
penalty function and $\rho^*(-Z_Q)=:\alpha_{min}(Q)$ for
$Q\in\mathcal Q$ is the minimal penalty function on $\mathcal Q$
that represents $\rho$ (see \cite{FoeSch04}).
\\
The penalty function $\alpha$ describes how seriously the
probabilistic model $Q\in\mathcal Q$ is taken. The value of the
convex risk measure $\rho(X)$ is the worst case of the expected loss
$E^{Q}[-X]$ reduced by $\alpha(Q)$, taken over all models
$Q\in\mathcal Q$ (\cite{FoeSch04}, Section~3.4).
\end{remark}
\noindent Coherent risk measures, as introduced in
\citet{ArtDelEbeHea99} and \citet{Delbaen02}, are convex risk
measures that are additionally positive homogeneous. For coherent
risk measures the dual representation (\ref{dkrm1}) reduces to:
$\rho$ is a lower semicontinuous coherent risk measure if and only
if there exists a non-empty subset of probability measures
$\widetilde{\mathcal{Q}}$ of $\mathcal{Q}$ with
$\{Z_Q:Q\in\widetilde{\mathcal{Q}}\}$ convex and weakly* closed in
$L^\infty$, such that
 \begin{equation}
    \label{dkohrm}\rho\of{X}=\sup_{Q\in \widetilde{\mathcal{Q}}}E^Q[-X].
 \end{equation}
This conditions is satisfied if $\rho$ admits (\ref{dkrm1}) with
$\mathcal A$ being a cone, which implies that $\sup_{X\in\mathcal
A}E^Q[-X]=\mathcal I_{\widetilde{\mathcal{Q}}}(Q)$ is an indicator
function equal to zero for $Q\in\widetilde{\mathcal{Q}}$ and
$+\infty$ otherwise.

\section{Decomposition of the Dynamic Problem}\label{sec3}
The dynamic optimization problem (\ref{dynop1}), (\ref{dynop2}) can
be split into the following two problems:
\begin{itemize}
\item[1.] Static optimization problem: Find an optimal modified
claim $\widetilde{\varphi}H$, where $\widetilde{\varphi}$ is a
randomized test solving
\begin{equation}
    \label{statop}
    \min_{{\varphi\in R_{0}}}\rho\of{(\varphi-1)H},
\end{equation}
\begin{equation}
    \label{statop2}
    R_{0}=\{\varphi: \; \Omega \to[0,1],\; \mathcal{F}_{T}-\mbox{measurable},\;
    \sup_{P^{*}\in \mathcal{P}_\sigma}E^{P^{*}}[\varphi H]\leq\widetilde{V}_{0}\}
\end{equation}
\item[2.] Representation problem: Find a superhedging strategy for
the modified claim $\widetilde{\varphi}H$.
\end{itemize}
This idea was introduced by \citet{FoeLeu99,FoeLeu00} using the
expectation of a loss function as risk measure and was used for
coherent risk measures in
\citet{Nakano03,Nakano04,Rudloff05CohH,Rudloff06Diss} analogously.
We obtain the following theorem for convex risk measures:

\begin{theorem}\label{Theorem Nakano}
Let $\widetilde{\varphi}$ be a solution of the minimization
problem (\ref{statop}) and let
$(\widetilde{V}_{0},\widetilde{\xi})$ be the admissible strategy,
where $\widetilde{\xi}$ is the superhedging strategy of the claim
$\widetilde{\varphi}H$. Then the strategy
$(\widetilde{V}_{0},\widetilde{\xi})$ solves the optimization
problem (\ref{dynop1}), (\ref{dynop2}) and it holds
\begin{equation}
    \label{dyn=stat}
    \min_{\scriptstyle{(V_{0},\xi)}}\rho(-(H-V_T)^+)
    =\min_{{\varphi\in R_{0}}}\rho\of{(\varphi-1)H}.
\end{equation}
\end{theorem}

\begin{proof}
Let $(V_0,\xi)$ with $V_0\leq\widetilde{V}_0$ be an admissible
strategy. We define the corresponding success ratio $
\varphi=\varphi_{(V_0,\xi)}$ as
\[
    \varphi_{(V_0,\xi)}:=I_{\{V_T\geq H\}}+\frac{V_T}{H}I_{\{V_T<H\}},
\]
where $I_{A}(\omega)$ is the stochastic indicator function equal to
one for $\omega\in A$ and zero otherwise. Thus,
$-(H-V_T)^+=(\varphi-1)H$. Since $V_t$ is a $\mathcal
P_\sigma$-supermartingale (\cite{DelSchach06}, Theorem~14.5.5) and
$\varphi H\leq V_T$:
\[
    \forall P^*\in\mathcal P_\sigma:\quad E^{P^*}[\varphi H]\leq E^{P^*}[V_T]\leq V_0\leq
    \widetilde{V}_0,
\]
hence, $\varphi\in R_0$. Thus,
\begin{equation}
    \label{sucess ratio}
    \rho(-(H-V_T)^+)=\rho((\varphi-1)H)\geq
    \rho((\widetilde{\varphi}-1)H),
\end{equation}
where $\widetilde{\varphi}$ is the solution to the static
optimization problem (\ref{statop}).  Consider the admissible
strategy $(\overline{V}_0, \widetilde{\xi})$, where
$\widetilde{\xi}$ is the superhedging strategy for the modified
claim $\widetilde{\varphi}H$ and
$\overline{V}_{0}\in[\widetilde{U}_{0},\widetilde{V}_{0}]$, where
$\widetilde{U}_{0}=\sup_{P^{*}\in
\mathcal{P}_\sigma}E^{P^{*}}[\widetilde{\varphi}H]$ is the
superhedging price of the modified claim $\widetilde{\varphi}H$ (see
Theorem~14.5.20, \cite{DelSchach06}). Inequality (\ref{sucess
ratio}) is especially satisfied for the success ratio of the
admissible strategy $(\overline{V}_0, \widetilde{\xi})$. Thus,
\begin{equation}
    \label{*}
    \rho((\varphi_{(\overline{V}_0, \widetilde{\xi})}-1)H)\geq
    \rho((\widetilde{\varphi}-1)H).
\end{equation}
To show the revers inequality, let us consider
$\varphi_{(\overline{V}_0,\widetilde{\xi})}H=\min(\widetilde{V}_T,H)$,
where
$\widetilde{V}_T=\overline{V}_0+\int_0^T\widetilde{\xi}_sdS_s$.
Because of $\widetilde{U}_{0}+\int_0^T
    \widetilde{\xi}_sdS_s\geq \widetilde{\varphi}H$ (superhedging) and
    $\overline{V}_{0}\in[\widetilde{U}_{0},\widetilde{V}_{0}]$, it holds
\[
    \widetilde{V}_T=\overline{V}_0+\int_0^T \widetilde{\xi}_sdS_s
    \geq\widetilde{\varphi}H+\overline{V}_0-\widetilde{U}_0
    \geq\widetilde{\varphi}H.
\]
Thus, $\varphi_{(\overline{V}_0, \widetilde{\xi})}H\geq
\widetilde{\varphi}H$. Since the convex risk measure $\rho$ is
monotone, we obtain
\begin{equation*}
    \rho((\varphi_{(\overline{V}_0, \widetilde{\xi})}-1)H)\leq
    \rho((\widetilde{\varphi}-1)H).
\end{equation*}
Together with (\ref{*}), we see that $\varphi_{(\overline{V}_0,
\widetilde{\xi})}$ attains the minimum of the static optimization
problem (\ref{statop}). Due to (\ref{sucess ratio}), we now have
\[
    \min_{\scriptstyle{(V_{0},\xi)}}\rho(-(H-V_T)^+)\geq
    \rho(-(H-\widetilde{V}_T)^+).
\]
Hence, $(\overline{V}_0, \widetilde{\xi})$ with
$\overline{V}_{0}\in[\widetilde{U}_{0},\widetilde{V}_{0}]$ is the
strategy that attains the minimum in the dynamic optimization
problem (\ref{dynop1}), (\ref{dynop2}) and equation (\ref{dyn=stat})
holds true.
\end{proof}

\begin{remark}\label{Remark U0}
In the case of risk measures that allow the construction of
$\widetilde{\varphi}$ via the Neyman-Pearson lemma directly (cf.
\cite{FoeLeu99} and some special cases of \cite{FoeLeu00}), one
can see that $\widetilde{U}_{0}=\widetilde{V}_{0}$ since the
optimal test $\widetilde{\varphi}$ attains the bound
$\widetilde{V}_{0}$ in (\ref{statop2}).
\\
In Theorem \ref{theorem SP}, equation (\ref{6.2}) of this paper we
will show that in the case of convex hedging the bound
$\widetilde{V}_{0}$ is as well attained by the optimal test. Thus,
the optimal strategy is $(\widetilde{V}_{0},\widetilde{\xi})$.
\end{remark}

\noindent In the following section, we will consider the static
optimization problem (\ref{statop}). To solve it we improve the
method used in \cite{Rudloff06Diss}.

\section{The Static Optimization Problem}\label{sec4}

Now, we will show that there exists a solution
$\widetilde{\varphi}$ of the static optimization problem
(\ref{statop}) and derive necessary and sufficient optimality
conditions. Therefore, we will construct the dual problem of
(\ref{statop}), deduce a result about the structure of a solution
for the inner problem of the dual problem and then solve the whole
problem.

\subsection{The Primal Problem and the Plan of its Solution}\label{sec4.1}

We impose the following assumption that has to be satisfied
throughout the remaining part of this paper:

\begin{assumption}\label{assumption rho}
Let $\rho:\;L^{1}\to\R\cup\{+\infty\}$ be a lower semicontinuous
convex risk measure that is continuous and finite in some
$(\varphi_0 -1)H$ with $\varphi_0\in R_0$.
\end{assumption}

\begin{remark}
A convex risk measure $\rho:\;L^{1}\to\R\cup\{+\infty\}$ is (without
assuming lower semicontinuity) continuous in the interior of its
domain (extended Namioka Theorem, see \cite{RusSha06},
Proposition~3.1 or \cite{FriBia06}, Theorem~2). Especially, if
$\rho(X)<+\infty$ for all $X\in L^1$, a convex risk measure admits
the representation (\ref{dkrm1}) and is continuous. But for extended
real valued convex risk measures we still need the assumption of
lower semicontinuity to obtain representation (\ref{dkrm1}).
\end{remark}

\noindent
Let us consider the measurable space $(\mathcal P_\sigma,
\mathcal S)$, where $\mathcal S$ is the $\sigma$-algebra generated
by all subsets of $\mathcal P_\sigma$. We denote by $\Lambda_+$ the
set of finite measures on $(\mathcal P_\sigma, \mathcal S)$.

\noindent We give an overview over the procedure to solve the static
optimization problem:
\begin{itemize}
    \item[(i)] Prove the existence of a solution
            $\widetilde{\varphi}$ to the primal problem (\ref{statop}) (Theorem~\ref{Th. Existenz})
            \begin{equation*}
                p=\min_{{\varphi\in R_{0}}}\rho\of{(\varphi-1)H}
                =\min_{\varphi\in R_{0}}\{\sup_{Q\in \mathcal{Q}}\{E^{Q}[(1-\varphi)H]-\sup_{X\in\mathcal A}E^Q[-X]\}\}.
            \end{equation*}
    \item[(ii)] Deduce the dual problem to (\ref{statop}) by Fenchel duality:
            \begin{equation}
            \label{review}
                d=\sup_{Q\in \mathcal{Q}}\{\inf_{\varphi\in R_{0}}\{E^{Q}[(1-\varphi)H]-\sup_{X\in\mathcal A}E^Q[-X]\}\}
            \end{equation}
            and prove the validity of strong duality $p=d$ (Theorem~\ref{theorem dualproblem}).
            We obtain the existence of a dual solution and can
            show that the problem is a saddle point problem.
    \item[(iii)] Consider the inner problem of the dual problem (\ref{review})
            for an arbitrary $Q\in \mathcal{Q}$:
            \begin{equation}
            \label{ip view}
                p^i(Q):=\max_{\varphi\in R_0}E^Q[\varphi H].
            \end{equation}
            Prove the existence of a solution $\widetilde{\varphi}_{Q}$ to
            (\ref{ip view}) (Lemma~\ref{lemma existence ip}).
                Deduce the dual problem of (\ref{ip view}) by Fenchel duality:
                \begin{equation*}
                    d^i(Q)=\inf_{\lambda\in \Lambda_{+}}\Big\{\int\limits_{\Omega}
                    [HZ_Q-H\int\limits_{\mathcal P_\sigma}Z_{P^{*}}d\lambda]^{+}dP
                    +\widetilde{V}_{0}\lambda(\mathcal P_\sigma)\Big\}.
                \end{equation*}
                Prove the validity of strong duality $p^i(Q)=d^i(Q)$
                and deduce the necessary and sufficient structure of
                a solution $\widetilde{\varphi}_{Q}$ to the inner
                problem (\ref{ip view}) (Theorem~\ref{theorem opt Test NP}).
    \item[(iv)] Apply Theorem~\ref{theorem dualproblem} and~\ref{theorem opt Test NP} to the primal
            problem (\ref{statop}) and deduce the necessary and sufficient structure of
            a solution $\widetilde{\varphi}$ to (\ref{statop})
            (Theorem~\ref{theorem SP}).
\end{itemize}

\noindent The existence of a solution $\widetilde{\varphi}$ to the
static optimization problem (\ref{statop}) can be shown
analogously to \citet[Proposition~1.3]{Nakano04}, where coherent
risk measures were considered.

\begin{theorem}\label{Th. Existenz}
    There exists a $\widetilde{\varphi}\in R_0$ solving
    the static optimization problem (\ref{statop}) and $\rho((\widetilde{\varphi}-1)H)$ is finite.
\end{theorem}

\begin{proof}
The set of randomized tests $R=\{\varphi: \; \Omega \to[0,1],\;
\mathcal{F}_{T} -\mbox{measurable}\}$ is weakly* compact as a
weakly* closed subset of the weakly* compact unit sphere in
$L^\infty$ (\cite{DunSchw88}, Theorem~V.4.2, V.4.3). Since the map
$\varphi\mapsto \sup_{P^{*}\in
\mathcal{P}_\sigma}E^{P^{*}}[\varphi H]$ is lower semicontinuous
in the weak* topology, the constrained set $R_0$ is weakly*
closed, hence weakly* compact. Because of the lower semicontinuity
of $\varphi\mapsto \sup_{Q\in
\mathcal{Q}}\{E^{Q}[(1-\varphi)H]-\sup_{X\in\mathcal A}E^Q[-X]\}$
in the weak* topology, there exists a $\widetilde{\varphi}\in R_0$
solving (\ref{statop}). $\rho((\widetilde{\varphi}-1)H)$ is finite
since $\rho$ is assumed to be finite in some $(\varphi_0 -1)H$
with $\varphi_0\in R_0$ (Assumption \ref{assumption rho}).
\end{proof}

\begin{remark}\label{remark uniqueness}
For measures of risk that are strictly convex one can additionally
show that any two solutions coincide $P-a.s.$ on $\{\omega: H>0\}$
\citep[see][Proposition~3.1]{FoeLeu00}. A convex risk measure cannot
be strictly convex since the translation property of $\rho$
(Definition \ref{def krm} (ii)) and $\rho(0)=0$ imply the linearity
of $\rho$ on the one dimensional subspace of $L^{1}$ generated by
the random variable equal to $1$ a.s. (see \citet{Ham05} for further
properties of translative functions). This means that for convex
risk measures one can only show the existence, not the essential
uniqueness of the solution.
\end{remark}

\subsection{The Dual Problem}\label{secSP}

In this subsection, we will construct the dual problem of
(\ref{statop}) and prove the validity of strong duality. We obtain
the existence of a dual solution and show that the problem is a
saddle point problem.

\begin{theorem}\label{theorem dualproblem}
    Strong duality holds: The values of the primal problem (\ref{statop})
    and its dual problem are equal ($p=d$), where
    the dual problem of (\ref{statop}) is the
    following with value $d$
    \begin{equation}
        \label{dual} d=\sup_{Q\in \mathcal{Q}}\{\inf_{\varphi\in
        R_{0}}\{E^{Q}[(1-\varphi)H]-\sup_{X\in\mathcal A}E^Q[-X]\}\}.
    \end{equation}
    $(\widetilde{Z}_{Q},\widetilde{\varphi})$ is a saddle point of the
    functional $E^{Q}[(1-\varphi)H]-\sup_{X\in\mathcal A}E^Q[-X]$, where
    $\widetilde{\varphi}$ is the solution of (\ref{statop}) and
    $\widetilde{Z}_{Q}=\frac{d\widetilde{Q}}{dP}$ is the solution of
    (\ref{dual}). Thus,
    \begin{equation}
        \nonumber
        \min_{\varphi\in R_{0}}\{\max_{Q\in \mathcal{Q}}
        \{E^{Q}[(1-\varphi)H]-\sup_{X\in\mathcal A}E^Q[-X]\}\}
        =\max_{Q\in \mathcal{Q}}\{\min_{\varphi\in R_{0}}
        \{E^{Q}[(1-\varphi)H]-\sup_{X\in\mathcal A}E^Q[-X]\}\}.
    \end{equation}
\end{theorem}

\begin{proof}
Problem (\ref{statop}) can be rewritten as
\[
    p=\min_{\varphi\in L^{\infty}}\{\rho\of{(\varphi-1)H}+\mathcal I_{R_{0}}(\varphi)\}.
\]
We denote $f(\varphi):=\mathcal I_{R_{0}}(\varphi)$ and
$g(A\varphi):=\rho\of{A\varphi-H}=\rho\of{(\varphi-1)H}$, where
the linear and continuous operator $A:L^{\infty}\to L^{1}$ is
defined by $A\varphi:=H\varphi$. The Fenchel dual problem is (see
\cite{EkeTem76}, Chapter~III, equation~(4.18))
\begin{equation}
    \label{du}
    d=\sup_{Y\in\,L^{\infty}}\{-f^{*}(A^{*}Y)-g^{*}(-Y)\},
\end{equation}
where $A^*$ is the adjoined operator of $A$ and $f^*, g^*$ are the
conjugate functions of $f$ and $g$, respectively. The value $p$ of
the primal problem is finite (Theorem~\ref{Th. Existenz}). The
function $f:\;L^{\infty}\to\R\cup\{+\infty\}$ is convex because of
the convexity of $R_{0}$. The function
$g:\;L^{1}\to\R\cup\{+\infty\}$ is convex since $\rho$ is convex.
Since $\rho$ is assumed to be continuous and finite in some
$(\varphi_0 -1)H$ with $\varphi_{0}\in R_{0}$
(Assumption~\ref{assumption rho}) we have strong duality $p=d$
(Theorem~III.4.1 and Remark~III.4.2 in \cite{EkeTem76}).
\\
The adjoined operator $A^*:L^{\infty}\to ba(\Omega,\mathcal{F},P)$
of A has to satisfy by definition the following equations:
\begin{eqnarray}
    \label{selbstadjOperator}\forall\, Y\in L^{\infty},\forall
    \varphi\in L^{\infty}:\quad
    \langle A^{*}Y,\varphi\rangle=\langle Y,A\varphi\rangle=E[\varphi HY].
\end{eqnarray}
To establish the dual problem, we calculate the conjugate
functions $f^{*}$ and $g^{*}$. With (\ref{selbstadjOperator}), we
obtain
\begin{eqnarray*}
    f^{*}(A^{*}Y)=\sup_{\varphi\in L^{\infty}}\{\langle A^{*}Y,\varphi\rangle-f(\varphi)\}
    =\sup_{\varphi\in R_{0}}E[\varphi HY].
\end{eqnarray*}
The function $g$ is defined by $g(X)=\rho(X-H)$. Its conjugate
function $g^*:\;L^{\infty}\to\R\cup\{+\infty\}$ is
\citep[Theorem~2.3.1 (vi)]{Zalinescu02}:
\begin{eqnarray*}
    g^{*}(Y)&=&\rho^{*}(Y)+\langle Y,H\rangle.
\end{eqnarray*}
Since $\dom \rho^*\subseteq\{-Z_Q:Q\in\mathcal Q\}$ and
$\rho^*(Y)=\sup_{X\in\mathcal A}E[XY]$ for $Y\in L^\infty$ with
$E[Y]=-1$ (see equation (\ref{r*})), the dual problem (\ref{du})
with value $d$ is
\begin{equation}
    \label{dualstruktur}d=\sup_{Q\in \mathcal{Q}}\{\inf_{\varphi\in
    R_{0}}\{E^{Q}[(1-\varphi)H]-\sup_{X\in\mathcal A}E^Q[-X]\}\}.
\end{equation}
The existence of a solution $\widetilde{Z}_{Q}$ to the dual
problem follows from the validity of strong duality \citep[see
Theorem~III.4.1,][]{EkeTem76}. Let $\widetilde{\varphi}$ be the
solution to the primal problem (\ref{statop}) (see Theorem
\ref{Th. Existenz}). Since
\begin{eqnarray*}
    p&=&\sup_{Q\in\mathcal{Q}}\{E^{Q}[(1-\widetilde{\varphi})H]-\sup_{X\in\mathcal A}E^Q[-X]\}\geq
    E^{\widetilde Q}[(1-\widetilde{\varphi})H]-\sup_{X\in\mathcal A}E^{\widetilde Q}[-X],
    \\
    d&=&\inf_{\varphi\in R_{0}}\{E^{\widetilde Q}[(1-\varphi)H]-\sup_{X\in\mathcal A}E^{\widetilde Q}[-X]\}\leq
    E^{\widetilde Q}[(1-\widetilde{\varphi})H]-\sup_{X\in\mathcal A}E^{\widetilde Q}[-X]
\end{eqnarray*}
and because of strong duality, we have
\[
 E^{\widetilde Q}[(1-\widetilde{\varphi})H]-\sup_{X\in\mathcal A}E^{\widetilde Q}[-X]\leq
  p=d\leq E^{\widetilde Q}[(1-\widetilde{\varphi})H]-\sup_{X\in\mathcal A}E^{\widetilde Q}[-X].
\]
Hence,
\begin{equation*}
    \min_{\varphi\in R_{0}}\{\max_{Q\in \mathcal{Q}}\{E^{Q}[
  (1-\varphi)H]-\sup_{X\in\mathcal A}E^Q[-X]\}\}=\max_{Q\in \mathcal{Q}}\{\min_{\varphi\in
    R_{0}}\{E^{Q}[(1-\varphi)H]-\sup_{X\in\mathcal A}E^Q[-X]\}\}.
\end{equation*}
Thus, $(\widetilde{Z}_{Q},\widetilde{\varphi})$ is a saddle point
of the function $E^{Q}[(1-\varphi)H]-\sup_{X\in\mathcal
A}E^Q[-X]$.
\end{proof}

\subsection{The Inner Problem of the Dual Problem}\label{sec4.3}
In this subsection, we consider the inner problem of
the dual problem (\ref{dual}) for an arbitrary, but fixed
$Q\in\mathcal{Q}$. We give a result about the structure of a
solution. This makes it possible to deduce a result about a saddle
point of Theorem \ref{theorem dualproblem} in our main theorem in
the next subsection. That means, we obtain a result about the
structure of a solution of the static optimization problem
(\ref{statop}).
\\
First let us consider the inner problem of the dual problem
(\ref{dual}) for a $Q\in\mathcal{Q}$ and let us denote with
$p^i(Q)$ its optimal value:
\begin{equation}
    \label{ip}
    p^i(Q):=\max_{\varphi\in R_0}E^{Q}[\varphi H].
\end{equation}

\begin{lemma}\label{lemma existence ip}
There exists a solution $\widetilde{\varphi}_{Q}$ to problem
(\ref{ip}) and $p^i(Q)$ is finite.
\end{lemma}

\begin{proof}
The assertion follows since $R_0$ is weakly* compact (see proof of
Theorem \ref{Th. Existenz}) and $\varphi\mapsto E^{Q}[\varphi H]$
is continuous in the weak* topology for all $Q\in\mathcal Q$.
\end{proof}

\begin{remark}
Problem (\ref{ip}) can be identified as a problem of test theory.
Let $R=\{\varphi: \; \Omega \to[0,1],\;
\mathcal{F}_{T}-\mbox{measurable}\}$ be the set of randomized
tests and let us define the measures $O$ and $O^{*}=O^{*}(P^{*})$
by $\frac{dO}{dQ}=H$ and $\frac{dO^{*}}{dP^{*}}=H$ for
$P^{*}\in{\mathcal{P}_\sigma}$. Problem (\ref{ip}) turns into
\begin{equation*}
    \max_{\varphi\in R}E^{O}[\varphi]
\end{equation*}
subject to
\begin{equation*}
    \forall P^{*}\in{\mathcal{P}_\sigma}:\quad E^{O^{*}}[\varphi ]\leq
    \widetilde{V}_{0}=:\alpha.
\end{equation*}
This is equivalent of looking for an optimal test
$\widetilde{\varphi}_{Q}$ when testing the compound hypothesis
$H_{0}=\{O^{*}(P^{*}):P^{*}\in{\mathcal{P}_\sigma}\}$, parameterized
by the class of equivalent sigma-martingale measures, against the
simple alternative hypothesis $H_{1}=\{O\}$ in a generalized sense.
In the generalized test problem \citep[Theorem~2.79]{Witting85}, $O$
and $O^{*}$ are not necessarily probability measures, but measures
and the significance level $\alpha$ is generalized to be a positive
continuous function $\alpha(P^{*})$. \citet{Witting85} deduced a
sufficient optimality condition for the optimal test
$\widetilde{\varphi}_Q$ and verified the validity of weak duality.
\end{remark}

\noindent We want to show that strong duality is satisfied. In this
case, the typical $0$-$1$-structure of $\widetilde{\varphi}_Q$ is
sufficient and necessary for optimality.
\\
We assign to (\ref{ip}) the following Fenchel dual problem and
denote by $d^i(Q)$ its optimal value
\begin{equation}
    \label{lambda}
    d^i(Q)=\inf_{\lambda\in \Lambda_{+}}\Big\{\int\limits_{\Omega}[HZ_{Q}-H\int\limits_{\mathcal
    P_\sigma}Z_{P^{*}}d\lambda]^{+}dP
    +\widetilde{V}_{0}\lambda(\mathcal P_\sigma)\Big\}.
\end{equation}
The following strong duality theorem holds true.

\begin{theorem}\label{theorem opt Test NP}
Strong duality holds true for problems (\ref{ip}) and
(\ref{lambda}), i.e.,
\[
        \forall Q\in\mathcal Q:\quad d^i(Q)=p^i(Q).
\]
Moreover, for each $Q\in\mathcal Q$ there exists a solution
$\widetilde{\lambda}_{Q}$ to (\ref{lambda}). The optimal
randomized test $\widetilde{\varphi}_{Q}$ of (\ref{ip}) has the
following structure:
\begin{eqnarray}
    \label{NP Loesung 1}\widetilde{\varphi}_{Q}(\omega)=\left\{\begin{array}{r@{\quad:\quad}l}
    1 & H Z_{Q}>H\int_{\mathcal{P}_\sigma}Z_{P^{*}}d\widetilde{\lambda}_{Q}(P^{*})
    \\[0.3cm]
    0 & H Z_{Q}<H\int_{\mathcal{P}_\sigma}Z_{P^{*}}d\widetilde{\lambda}_{Q}(P^{*})
    \end{array}\right. \quad
    P -a.s.
\end{eqnarray} and
\begin{equation}
   \label{NP Loesung 2}
   E^{P^{*}}[\widetilde{\varphi}_{Q}H]=\widetilde{V}_{0}\quad\quad\widetilde{\lambda}_{Q} -a.s.
\end{equation}
\end{theorem}

\begin{proof}
Let $\mathcal L$ be the linear space of all bounded and measurable
real functions on $(\mathcal{P}_\sigma, \mathcal S)$ with pointwise
addition, multiplication with real numbers and pointwise partial
order $l_1\leq l_2\Leftrightarrow l_2-l_1\in\mathcal
L_+:=\{l\in\mathcal L:\forall P^*\in {\mathcal{P}_\sigma}:l(P^*)\geq
0\}$. We recall that $\mathcal S$ is the $\sigma$-algebra generated
by all subsets of $\mathcal P_\sigma$.
\\
Let $\Lambda$ be the space of all $\sigma$-additive signed measures
on $(\mathcal P_\sigma, \mathcal S)$ of bounded variation. We regard
$\mathcal L$ and $\Lambda$ as the duality pair associated with the
bilinear form $\langle l, \lambda\rangle=\int_{\mathcal
P_\sigma}ld\lambda$ for $l\in\mathcal L$ and $\lambda\in\Lambda$,
see \citet[Theorem~13.5]{AliBor99}. We endow the space $\mathcal L$
with the Mackey topology $\tau(\mathcal L, \Lambda)$, which ensures
that the topological dual of $(\mathcal L,\tau(\mathcal L,
\Lambda))$ is $\Lambda$ and that $\mathcal L$ is a barrelled space
(\cite{HusKha78}, Corollary~II.2, II.4).
\\
We define a linear and continuous operator
$B:(L^\infty,\|\cdot\|_{L^\infty})\to(\mathcal L,\tau(\mathcal L,
\Lambda))$ by $(B\varphi)(P^*):=-E^{P^*}[H\varphi]$ for
$P^*\in{\mathcal{P}_\sigma}$. $B$ is continuous since for every
sequence $\varphi_n\rightarrow\varphi$ in
$(L^\infty,\|\cdot\|_{L^\infty})$, it holds that
$B\varphi_n\rightarrow B\varphi$ in $(\mathcal L,\|\cdot\|_{\mathcal
L})$, where $\|\l\|_{\mathcal L}:=\sup_{P^*\in
\mathcal{P}_\sigma}|l(P^*)|$, since
\[
    \sup_{P^*\in \mathcal{P}_\sigma}|B(\varphi_n-\varphi)(P^*)|\leq
    \|\varphi_n-\varphi\|_{L^\infty}U_0
\]
and $U_0<+\infty$ (inequality (\ref{SHP})). Thus, $B\varphi_n$
converges also in the weaker topology $\tau(\mathcal L, \Lambda)$.
We define the functions $\textbf{1},\textbf{0}\in \mathcal L$ by
\[
    \forall P^*\in {\mathcal{P}_\sigma}:\textbf{1}(P^*)=1\in\R,\;\textbf{0}(P^*)=0\in\R.
\]
Problem (\ref{ip}) is
\begin{equation*}
    \max_{\varphi\in R}E^{Q}[\varphi H],
\end{equation*}
\begin{equation}
    \label{NB}
    \forall P^{*}\in{\mathcal{P}_\sigma}:\quad E^{P^{*}}[\varphi H]\leq\widetilde{V}_{0}.
\end{equation}
The constraint (\ref{NB}) can be rewritten as
\[
    \widetilde{V_{0}}\textbf{1}+B\varphi\geq \textbf{0}
    \Leftrightarrow
    B\varphi\in \mathcal L_{+}-\widetilde{V_{0}}\textbf{1}.
\]
Then, we can write problem (\ref{ip}) equivalently as
\begin{equation}
    \label{-p neu}-p^i(Q)=\min_{\varphi\in L^{\infty}}\Big\{-E^{Q}[\varphi H]+\mathcal
    I_{R}(\varphi)
    +\mathcal I_{\mathcal
    L_{+}-\widetilde{V}_{0}\textbf{1}}(B\varphi)\Big\}.
\end{equation}
Let us define the functions $f(\varphi):=-E^{Q}[\varphi H]+\mathcal
I_{R}(\varphi)$ and $g(B\varphi):=\mathcal I_{\mathcal
L_{+}-\widetilde{V}_{0}\textbf{1}}(B\varphi)$ in (\ref{-p neu}). We
want to establish the dual problem of (\ref{-p neu}) as in
\citet{EkeTem76} (Chapter~III, equation~(4.18)):
\begin{equation}
    \label{-di}
    -d^i(Q)=\sup_{\lambda\in\Lambda}\Big\{-f^{*}(B^{*}\lambda)-g^{*}(-\lambda)\Big\}.
\end{equation}
The conjugate function of $g$ is
\begin{eqnarray*}
    g^{*}(\lambda)&=&\sup_{\widetilde{l}\in\mathcal L}\Big\{\langle \widetilde{l},\lambda\rangle-\mathcal I_{\mathcal
    L_{+}-\widetilde{V}_{0}\textbf{1}}(\widetilde{l})\Big\}
    =\sup_{\widetilde{l}\in\mathcal L_{+}-\widetilde{V}_{0}\textbf{1}}\langle \widetilde{l},\lambda\rangle
    =\sup_{l\in\mathcal L_{+}}\langle l-\widetilde{V}_{0}\textbf{1},\lambda\rangle
    \\
    &=&\sup_{l\in\mathcal L_{+}}\langle
    l,\lambda\rangle-\widetilde{V}_{0}\int\limits_{\mathcal P_\sigma}d\lambda
    =\mathcal I_{\mathcal L_{+}^{*}}(\lambda)-\widetilde{V}_{0}\lambda(\mathcal P_\sigma),
\end{eqnarray*}
where $\mathcal L_{+}^{*}$ is the negative dual cone of $\mathcal
L_{+}$. To establish the conjugate function of $f$
\begin{eqnarray*}
    f^{*}(B^{*}\lambda)&=&\sup_{\varphi\in L^{\infty}}\Big\{\langle
    B^{*}\lambda,\varphi\rangle+E^{Q}[\varphi H]-\mathcal
    I_{R}(\varphi)\Big\},
\end{eqnarray*}
we have to calculate $\langle B^{*}\lambda,\varphi\rangle$, where
$B^{*}:\;\Lambda \to ba(\Omega,\mathcal{F},P)$ is the adjoined
operator of $B$. By definition of $B^*$, the equation $\langle
B^{*}\lambda,\varphi\rangle=\langle\lambda,B\varphi\rangle$ has to
be satisfied for all $\varphi\in L^{\infty},\lambda\in \Lambda$ (see
\cite{AliBor99}, Definition 6.51). Thus,
\[
    \nonumber\forall\varphi\in L^{\infty},\forall\, \lambda\in\Lambda:\;
    \langle B^{*}\lambda,\varphi\rangle=\int_{\mathcal P_\sigma}-E^{P^{*}}[\varphi H]d\lambda.
\]
Hence the conjugate function of $f$ is
\[
     f^{*}(B^{*}\lambda)
    =\sup_{\varphi\in R}\Big\{-\int_{\mathcal P_\sigma}E^{P^{*}}[\varphi H]d\lambda+E^{Q}[\varphi H]\Big\}.
\]
The dual problem (\ref{-di}) becomes
\begin{eqnarray*}
    -d^i(Q)&=&\sup_{\lambda\in\Lambda}\Big\{-\sup_{\varphi\in R}\Big\{
    -\int_{\mathcal P_\sigma}E^{P^{*}}[\varphi H]d\lambda+E^{Q}[\varphi H]\Big\}-\mathcal I_{-\mathcal L_{+}^{*}}(\lambda)
    -\widetilde{V}_{0}\lambda(\mathcal P_\sigma)\Big\},
    \\
    d^i(Q)&=&\inf_{\lambda\in -\mathcal L_{+}^{*}}\Big\{\sup_{\varphi\in R}
    \Big\{-\int_{\mathcal P_\sigma}E^{P^{*}}[\varphi H]d\lambda+E^{Q}[\varphi H]\Big\}
    +\widetilde{V}_{0}\lambda(\mathcal P_\sigma)\Big\},
\end{eqnarray*}
where $-\mathcal L_{+}^{*}=\{\lambda\in\Lambda:\forall l\in \mathcal
L_{+} :\langle l,\lambda\rangle\geq 0\}$. It can be seen easily that
$-\mathcal L_{+}^{*}=\Lambda_{+}$ is the set of finite measures on
$(\mathcal P_\sigma, \mathcal S)$. Thus,
\begin{equation}
    \label{di2}d^i(Q)=\inf_{\lambda\in\Lambda_{+}}\Big\{\sup_{\varphi\in
    R}\Big\{-\int_{\mathcal P_\sigma}E^{P^{*}}[\varphi H]d\lambda+E^{Q}[\varphi H]\Big\}
    +\widetilde{V}_{0}\lambda(\mathcal P_\sigma)\Big\}.
\end{equation}
The spaces $(\Omega,\mathcal{F},P)$ and $(\mathcal P_\sigma,
\mathcal S,\lambda)$ for $\lambda\in \Lambda_{+}$ are positive,
finite measure spaces. Furthermore, the function
$f(\omega,P^*)=H(\omega)Z_{P^*}(\omega)\varphi(\omega)$ is
measurable for all $\varphi\in R$ and it holds that for all
$\lambda\in \Lambda_{+}$ and for all $\varphi\in R$
\[
  \int\limits_{\mathcal P_\sigma}\int\limits_\Omega |HZ_{P^*}\varphi| dP d\lambda
  \stackrel{\|\varphi\|_{L^\infty}\leq 1}{\leq}\sup_{P^{*}\in \mathcal P_\sigma}\|HZ_{P^*}\|_{L^1}\lambda(\mathcal P_\sigma)
  \stackrel{(1)}{<}+\infty.
\]
Thus, we can apply Tonelli's Theorem
\cite[Corollary~III.11.15]{DunSchw88} and obtain that the order of
integration can be changed, i.e., for all $\lambda\in \Lambda_{+}$
and for all $\varphi\in R$
\begin{equation*}
  \int\limits_{\mathcal P_\sigma}\int\limits_\Omega HZ_{P^*}\varphi dP d\lambda
  =\int\limits_\Omega\int\limits_{\mathcal P_\sigma}HZ_{P^*}\varphi d\lambda dP<+\infty.
\end{equation*}
Since in (\ref{di2}) only elements $\lambda\in \Lambda_{+}$ and
$\varphi\in R$ have to be considered, we can change the order of
integration and obtain
\begin{equation}
    \label{d}
    d^i(Q)=\inf_{\lambda\in\Lambda_{+}}\Big\{\sup_{\varphi\in
    R}E[\varphi(HZ_Q-H\int_{\mathcal P_\sigma}Z_{P^*}d\lambda)]
    +\widetilde{V}_{0}\lambda(\mathcal P_\sigma)\Big\}.
\end{equation}
Since $\varphi\in R$ is a randomized test, the supremum over all
$\varphi\in R$ in (\ref{d}) is attained by
\[
    \overline{\varphi}(\omega)=\left\{\begin{array}{r@{\quad:\quad}l}
    1 & HZ_{Q}>H\int_{\mathcal{P}_\sigma}Z_{P^{*}}d\lambda
    \\[0.3cm]
    0 & HZ_{Q}<H\int_{\mathcal{P}_\sigma}Z_{P^{*}}d\lambda
    \end{array}\right. \quad
    P -a.s.
\]
If we denote $HZ_{Q}-H\int\limits_{\mathcal
P_\sigma}Z_{P^{*}}d\lambda=:\nu_{\lambda}(\omega)$, the value of
the dual
    problem is
\begin{equation}\label{di}
    d^i(Q)=\inf_{\lambda\in \Lambda_{+}}\Big\{\int\limits_{\Omega}\nu_{\lambda}^{+}(\omega)dP
    +\widetilde{V}_{0}\lambda(\mathcal P_\sigma)\Big\}.
\end{equation}
This is equation (\ref{lambda}) of Theorem \ref{theorem opt Test
NP}. To verify the validity of strong duality we have to use a
weaker regularization condition than used in Theorem~\ref{theorem
dualproblem}. In \citet[Theorem~5]{BorZhu06} it is shown that strong
duality holds if $f$ and $g$ are convex and lower semicontinuous and
if there exists some $\varphi_{0}\in\dom f$ such that
$B\varphi_{0}\in\core(\mathcal L_{+}-\widetilde{V}_{0}\textbf{1})$,
where $\core(M)$ is the algebraic interior of a set $M$. If we take
$\varphi_{0}\equiv 0,\; 0\in\dom f$, we have to show that
$B\varphi_{0}=\textbf{0}\in\core(\mathcal
L_{+}-\widetilde{V}_{0}\textbf{1})$ for $\widetilde{V}_{0}>0$. This
holds true since for all $l\in\mathcal L$ there exists a $s=s(l)>0$,
such that for all $0\leq t\leq s(l)$ it holds
$\textbf{0}+tl\in(\mathcal L_{+}-\widetilde{V}_{0}\textbf{1})$. We
can choose $s(l)=\widetilde{V}_{0}/\|l\|_{\mathcal L}>0$ for $l\neq
\textbf{0}$ and $s(l)=c>0$ for $l=\textbf{0}$, $c>0$ arbitrary,
where $\|l\|_{\mathcal L}:=\sup_{P^*\in\mathcal P_\sigma}|l(P^*)|$.
\\
The function $f$ is convex and lower semicontinuous since the set
$R$ is convex and closed. The function $g$ is convex since the set
$(\mathcal L_{+}-\widetilde{V}_{0}\textbf{1})$ is convex and it is
lower semicontinuous w.r.t. the Mackey topology $\tau(\mathcal L,
\Lambda)$ if and only if the set $(\mathcal
L_{+}-\widetilde{V}_{0}\textbf{1})$ is closed w.r.t. the Mackey
topology $\tau(\mathcal L, \Lambda)$. To show this, we use that a
convex set is closed w.r.t. the Mackey topology $\tau(\mathcal L,
\Lambda)$ if and only if it is closed w.r.t. the weak topology
$\sigma(\mathcal L, \Lambda)$. Take a net $l_\alpha$ in $(\mathcal
L_{+}-\widetilde{V}_{0}\textbf{1})$ that converges weakly to $l$.
Thus, for all $\lambda\in\Lambda$ it holds $\int_{\mathcal
P_\sigma}l_\alpha d\lambda\rightarrow\int_{\mathcal
P_\sigma}ld\lambda$ and for all $P^*\in\mathcal P_\sigma$ and all
$\alpha$ it holds $l_\alpha(P^*)+\widetilde{V}_{0}\geq 0$. Suppose
there exists a $\overline{\mathcal P}\in\mathcal S$ and a
$\lambda\in\Lambda_+$ with $\lambda(\overline{\mathcal P})>0$ and
$\lambda(\mathcal P_\sigma\backslash\overline{\mathcal P})=0$ such
that $l(P^*)+\widetilde{V}_{0}<0$ for all $P^*\in\overline{\mathcal
P}$. Then, $\int_{\mathcal
P_\sigma}(l(P^*)+\widetilde{V}_{0})d\lambda<0$, which is a
contradiction to $\int_{\mathcal
P_\sigma}(l(P^*)+\widetilde{V}_{0})d\lambda=\lim_{\alpha}\int_{\mathcal
P_\sigma}(l_\alpha(P^*)+\widetilde{V}_{0}) d\lambda\geq 0$ for all
$\lambda\in\Lambda_+$. Since $\mathcal S$ contains also the
one-point subsets of $\mathcal P_\sigma$, it follows that
$l\in(\mathcal L_{+}-\widetilde{V}_{0}\textbf{1})$. Thus, strong
duality holds true.
\\
To demonstrate the dependence from the selected measure
$Q\in\mathcal{Q}$ we use the notation $\widetilde{\varphi}_Q$ and
$\widetilde{\lambda}_Q$ for the primal and dual solution,
respectively. The existence of a solution $\widetilde{\varphi}_Q\in
R_0$ of the primal problem was proved in Lemma~\ref{lemma existence
ip}. Now with strong duality the existence of a dual solution
$\widetilde{\lambda}_Q$ follows and the values of the primal and
dual objective function at $\widetilde{\varphi}_Q$, respectively
$\widetilde{\lambda}_Q$, coincide. This leads to a necessary and
sufficient optimality condition.
\\We consider the primal objective function
\begin{eqnarray*}
    E[\varphi H Z_{Q}]&=&\int\limits_{\Omega}\varphi H Z_{Q}dP
    \\
    &=&\int\limits_{\Omega}\varphi\Big[H Z_{Q}-H\int\limits_{\mathcal{P}_\sigma}Z_{P^{*}}d\lambda\Big]dP
    +\int\limits_{\mathcal{P}_\sigma}\int\limits_{\Omega}\varphi HZ_{P^{*}}dP d\lambda
    \\
    &=&\int\limits_{\Omega}\varphi\nu_{\lambda}^{+}(\omega) dP-\int\limits_{\Omega}\varphi\nu_{\lambda}^{-}(\omega) dP
    +\int\limits_{\mathcal{P}_\sigma}\int\limits_{\Omega}\varphi HZ_{P^{*}}dP d\lambda
\end{eqnarray*}
and subtract it from the dual objective function. Because of
strong duality the difference has to be zero at
$\widetilde{\varphi}_Q$, respectively $\widetilde{\lambda}_Q$:
\[
    \int\limits_{\Omega}\Big[1-\widetilde{\varphi}_Q\Big]\nu_{\widetilde{\lambda}_Q}^{+}(\omega) dP
    +\int\limits_{\Omega}\widetilde{\varphi}_Q\nu_{\widetilde{\lambda}_Q}^{-}(\omega) dP
    +\int\limits_{\mathcal{P}_\sigma}\Big[\widetilde{V}_{0}-\int\limits_{\Omega}\widetilde{\varphi}_Q HZ_{P^{*}}dP\Big]
    d\widetilde{\lambda}_Q=0.
\]
The sum of these three nonnegative integrals is zero if and only
if $\widetilde{\varphi}_Q\in R_0$ satisfies condition (\ref{NP
Loesung 1}) and (\ref{NP Loesung 2}) of Theorem \ref{theorem opt
Test NP}. To stress that $\lambda$ is a measure on $\mathcal
P_\sigma$, we use in Theorem \ref{theorem opt Test NP} the
notation $\lambda(P^{*})$.
\end{proof}
\noindent
For each $Q\in\mathcal{Q}$ there exist a primal and a
dual solution $\widetilde{\varphi}_{Q}$,
$\widetilde{\lambda}_{Q}$, respectively. If $Q=\widetilde{Q}$ is
the solution of the outer problem of (\ref{dual}),
$\widetilde{\varphi}_{\widetilde{Q}}$ is the solution of the
static optimization problem (\ref{statop}).

\subsection{The Saddle Point}\label{sec4.4}
Now, let us consider the saddle point problem described in Theorem
\ref{theorem dualproblem}. With Theorem \ref{theorem opt Test NP} it
follows that
\begin{eqnarray}
    \nonumber&\max\limits_{Q\in \mathcal{Q}}&\hspace{-0.29cm}\min\limits_{\varphi\in
    R_{0}}\{E^{Q}[(1-\varphi)H]-\sup_{X\in\mathcal A}E^Q[-X]\}
    = \max\limits_{Q\in \mathcal{Q}}\{E^{Q}[H]-p^i(Q)-\sup_{X\in\mathcal A}E^Q[-X]\}
    \\
    \nonumber&=& \hspace{-0.32cm}\max\limits_{Q\in \mathcal{Q}}\{E^{Q}[H]-d^i(Q)-\sup_{X\in\mathcal A}E^Q[-X]\}
    \\
    \nonumber&=& \hspace{-0.32cm}\max\limits_{Q\in \mathcal{Q}}\Big\{\max\limits_{\lambda\in \Lambda_{+}}
    \Big\{-E^P[(HZ_{Q}-H\int\limits_{\mathcal P_\sigma}Z_{P^{*}}d\lambda)^{+}]
    -\widetilde{V}_{0}\lambda(\mathcal P_\sigma)\Big\}+E^{Q}[H]-\sup_{X\in\mathcal A}E^Q[-X]\Big\}
    \\
    \nonumber&=& \hspace{-0.32cm}\max\limits_{Q\in \mathcal{Q}, \lambda\in \Lambda_{+}}
    \Big\{E^P[HZ_{Q}\wedge H\int\limits_{\mathcal P_\sigma}Z_{P^{*}}d\lambda]-\widetilde{V}_{0}\lambda(\mathcal P_\sigma) -\sup_{X\in\mathcal A}E^Q[-X]\Big\},
\end{eqnarray}
where $x\wedge y=min(x,y)$. With Theorem \ref{theorem dualproblem}
it follows that $\widetilde{Q}$ attains the maximum w.r.t. $Q\in
\mathcal{Q}$. Theorem \ref{theorem opt Test NP} shows the
existence of a
$\widetilde{\lambda}=\widetilde{\lambda}_{\widetilde{Q}}$ that
attains the maximum w.r.t. $\lambda\in \Lambda_{+}$. Thus, there
exists a pair $(\widetilde{Q},\widetilde{\lambda})$ solving
\begin{equation}
    \label{opt pair}\max_{Q\in \mathcal{Q}, \lambda\in \Lambda_{+}}
    \Big\{E^P[HZ_{Q}\wedge H\int\limits_{\mathcal P_\sigma}Z_{P^{*}}d\lambda]-\widetilde{V}_{0}\lambda(\mathcal P_\sigma) -\sup_{X\in\mathcal A}E^Q[-X]\Big\}.
\end{equation}
Now, our main theorem follows.

\begin{theorem}\label{theorem SP}
Let $(\widetilde{Q},\widetilde{\lambda})$ be the optimal pair in (\ref{opt pair}).
\begin{itemize}
    \item   The solution of the static optimization problem (\ref{statop}) is
        \begin{eqnarray}
            \label{6.1}\widetilde{\varphi}(\omega)=\left\{\begin{array}{r@{\quad:\quad}l}
            1 & H \widetilde{Z}_{Q}>H\int_{\mathcal{P}_\sigma}Z_{P^{*}}d\widetilde{\lambda}(P^{*})
            \\[0.3cm]
            0 & H \widetilde{Z}_{Q}<H\int_{\mathcal{P}_\sigma}Z_{P^{*}}d\widetilde{\lambda}(P^{*})
            \end{array}\right. \quad P -a.s.
        \end{eqnarray}
            with
        \begin{equation}
            \label{6.2}E^{P^{*}}[\widetilde{\varphi}H]=\widetilde{V}_{0}\quad\quad
        \widetilde{\lambda} -a.s.
        \end{equation}
    \item   $(\widetilde{\varphi},\widetilde{Z}_{Q})$ is the saddle point of Theorem \ref{theorem dualproblem}.
    \item   $(\widetilde{V}_{0}, \widetilde{\xi})$ solves the dynamic convex hedging problem (\ref{dynop1}), (\ref{dynop2}), where
            $ \widetilde{\xi}$ is the superhedging strategy of the modified claim $\widetilde{\varphi}H$.
\end{itemize}
\end{theorem}

\begin{proof}
The results follow from Theorem \ref{theorem dualproblem},
\ref{theorem opt Test NP} and \ref{Theorem Nakano}.
\end{proof}

\begin{remark}
It follows that there exists a $[0,1]$-valued random variable
$\delta$ such that $\widetilde{\varphi}$ as in Theorem \ref{theorem
SP} satisfies
\begin{equation}
    \label{delta}
   \widetilde{\varphi}(\omega)=I_{\{H \widetilde{Z}_{Q}>H\int_{\mathcal{P}_\sigma}Z_{P^{*}}d\widetilde{\lambda}(P^{*})\}}(\omega)
   +\delta(\omega)I_{\{H \widetilde{Z}_{Q}=H\int_{\mathcal{P}_\sigma}Z_{P^{*}}d\widetilde{\lambda}(P^{*})\}}(\omega),
\end{equation} where $I_{A}(\omega)$ is the stochastic indicator
function equal to one for $\omega\in A$ and zero otherwise. $\delta$
has to be chosen such that $\widetilde{\varphi}$ satisfies
(\ref{6.2}).
\\
We do not have uniqueness of a solution $\widetilde{\varphi}$ (see
Remark~\ref{remark uniqueness}). If we for instance choose $\delta$
to be constant, equations (\ref{6.2}) and (\ref{delta}) lead to one
particular $\delta$ and thus to one particular solution
$\widetilde{\varphi}$.
\end{remark}

\begin{remark}
From equation (\ref{6.2}) it follows, that (except in the case where
$\widetilde{\lambda}$ takes only the value zero) the superhedging
price of the modified claim $\widetilde{\varphi}H$ is equal to the
capital boundary $\widetilde{V}_{0}$. Then, $\widetilde{V}_{0}$ is
the minimal amount of capital that is necessary to solve together
with $\widetilde{\xi}$ the dynamic problem (\ref{dynop1}),
(\ref{dynop2}).
\end{remark}

\noindent To summarize, the admissible strategy that minimizes the
convex shortfall risk consists in superhedging a modified claim
$\widetilde{\varphi}H$ that has the form of a knock-out option.
\bibliographystyle{techmech}
\bibliography{DissRefConH}

\end{document}